\documentclass[runningheads]{llncs}
\makeatletter
\newcommand{\@chapapp}{\relax}%
\makeatother

\usepackage[title,toc,titletoc,header]{appendix}
\usepackage{graphicx}
\usepackage{algorithmic}
\usepackage{extarrows}
\usepackage[shortlabels]{enumitem}
\usepackage{caption}
\usepackage{subcaption}
\usepackage{floatrow}
\usepackage{blindtext}
\usepackage{verbatim}
\usepackage{xcolor}
\usepackage{cancel}
\usepackage{enumerate}
\usepackage{array}
\usepackage{multirow}
\usepackage{slashbox}
\usepackage{diagbox}
\usepackage{listings}
\usepackage{hyperref}
\usepackage{caption}
\usepackage{makeidx}  % allows for indexgeneration
\usepackage{booktabs}
\usepackage{graphicx}
\usepackage{amssymb}
\usepackage{amsmath}
\usepackage{extarrows}
\usepackage{epstopdf}
\usepackage{blindtext}
\usepackage{color}
\usepackage{verbatim}
\usepackage{datetime}
\usepackage{textcomp}
\def\Reals{\mathbb{R}}

\newcommand{\Real}{{\mathbb R}}

\begin{document}
  \title{Rate Lifting for Stochastic Process Algebra\\
	  -- Exploiting Structural Properties --}
\titlerunning{Rate Lifting for SPA}
\author{Markus Siegle \and Amin Soltanieh}
%version of \today
\authorrunning{ M. Siegle, A. Soltanieh}
\institute{Department of Computer Science, Bundeswehr University Munich
\email{markus.siegle@unibw.de, amin.soltanieh@unibw.de}\\
}
\maketitle              % typeset the header of the contribution

\begin{abstract}

This report presents an algorithm for determining the unknown rates in the sequential processes of a Stochastic Process Algebra (SPA) model,
provided that the rates in the combined flat model are given.
Such a rate lifting is useful for model reengineering and model repair.
Technically, the algorithm works by solving systems of nonlinear equations
and -- if necessary -- adjusting the model's synchronisation structure without changing its transition system.
This report contains the complete pseudo-code of the algorithm.
The approach taken by the algorithm exploits some structural properties of SPA systems, which are formulated here for the first time
and could be very
beneficial also in other contexts.

\end{abstract}

\section{Introduction}
\label{sec:Intro}

Stochastic Process Algebra (SPA) is a family of formalisms widely used
in the area of quantitative modelling and evaluation.
Typical members of this family are
PEPA \cite{HillstonBook},
TIPP \cite{Goetz94},
EMPA \cite{Bernardo99},
CASPA \cite{kuntz:04b},
but also the reactive modules language of tools such as PRISM \cite{KNP11} and STORM \cite{storm:2017}.
Originally devised for classical performance and dependability modelling,
SPA models are now frequently used in probabilistic
model checking projects.

This paper presents a solution to the following problem:
Given a compositional SPA specification
where the transition rates of its components are unknown,
but given all transition rates of the associated low-level,
flat transition system,
find the unknown transition rates for the components
of the high-level SPA model.
An alternative formulation of the same problem is for a compositional
SPA specification with known original transition rates in its components,
but given rate modification factors for (a subset of) the transition rates
in its flat low-level model.
Here the task is to find new transition rates for the components of the
high-level SPA model,
such that the resulting rates in the flat model will be modified as desired.
The first formulation is from the perspective of systems reengineering
(to be more specific, one could call it rate reengineering),
whereas the second one pertains to model checking and model repair \cite{Bartocci:2011,Chen:2013,Pathak:2015}.
We will refer to both variants of the problem as ``rate lifting problem''.

An algorithm that solves the rate lifting problem for SPA models
with $n=2$ components was presented in \cite{Soltanieh:2020},
the equation system involved being studied in \cite{Soltanieh:2021}.
However,
developing a rate lifting algorithm for a general number $n \geq 3$ of processes
turns out to be a much bigger challenge, since -- firstly -- SPA models with
$n$ components may have a much more complex synchronisation structure
than for $n=2$,
and it is the synchronisation structure which plays an essential
role during the execution of the algorithm.
Secondly, components of SPA models may contain selfloops
(meant to synchronise with other components),
and -- related to this -- 
the transition system underlying
a compositional SPA model is actually
a flattened multi-transition system~\cite{HillstonBook,Goetz94}.
These two facts have to be considered during the
necessary deconstruction of a flat transition, and they strongly contribute
to the complication of the problem.
So, in this paper we develop a rate lifting algorithm for an SPA
system consisting of $n$ components, where $n$ is arbitrary.
The algorithm will assign (new) values to the components'
transition rates and -- under certain circumstances --
it will change the synchronisation structure of the SPA model.
The latter means that the algorithm may add actions to certain synchronisation sets
and in consequence it will insert additional selfloops at some
specific component states, but it will do this in such a way that
the set of reachable states and the set of transitions of the overall
model are not changed.
Only the transition rates of the overall model are set/changed as desired.
Technically, the algorithm works by setting up and solving systems of nonlinear (actually multilinear) equations.

It is quite easy to see that an arbitrary assignment of rates to the transitions
of the low-level transition system may not always be realisable by
suitable rates in the components, i.e.\ not every instance of the rate
lifting problem has a solution.
Therefore, naturally, the algorithm presented in this paper will
not always succeed.
However, it is guaranteed that the algorithm will find a solution,
if such a solution exists
(see Sec.~\ref{AppendixCorrectOptimal}).

We build our algorithm based on certain structural properties of SPA systems,
which can be exploited in the course of the algorithm.
As an example, for a given transition in one of the SPA components,
it is necessary to identify the partners
which may or must synchronise with it.
To the best of our knowledge,
these fundamental properties have not previously been addressed
in the literature,
which is suprising, since they could be very valuable also in other contexts.
For example, in compositional system verification,
distinguishing between different types of neighbourhoods of processes or
determining the participating set of a transition
(see Sec.~\ref{section:StructProps})
is the key to establishing 
dependence / independence relations between processes.

\section{Structural Properties of SPA}
\label{section:StructProps}
We consider a simple but fairly general class of Markovian Stochastic Process
Algebra models constructed by the following grammar:

\begin{definition}
{\em (SPA language)}
For a finite set of actions $Act$,
%(including the internal action $\tau$),
let $a \in Act$
and $A \subseteq Act$.
Let $\lambda \in \Real^{>0}$ be a transition rate.
%and  let $X \in Var$ be a process variable.
An SPA system $Sys$ is a process of type $Comp$, constructed according to the following grammar:
\begin{eqnarray*}
Comp \; & := & \; (Comp \; ||_A \; Comp) \; \big| \; Seq\\
Seq \; & := & \; 0 \; \big| \; \; (a, \lambda);Seq  \; \big| \; Seq + Seq
\; \big| \; V
\end{eqnarray*}

\end{definition}
$Seq$ stands for sequential processes, and $Comp$ for composed processes.
$V$ stands for a process variable for a sequential process,
which can be used to define cyclic behaviour (including selfloops).
One could add a recursion operator, the special invisible action $\tau$,
hiding and other features, but this is not essential for our purpose.
The semantics is standard, i.e.\ the SPA specification is mapped to the underlying flat transition system (an action-labelled CTMC), see e.g.\ \cite{HillstonBook,Goetz94}.
It should be mentioned that we assume multiway synchronisation\footnote{Unlike, e.g., the process algebra CCS which has two-way synchronisation \cite{Milner:1980}},
i.e.\ the synchronisation of two
$a$-transitions yields another $a$-transition
(whose rate is a function of the two partner transitions, or
-- even more general -- of the two partner processes),
which can then participate in further $a$-synchronisations, etc..

An SPA system corresponds to a process tree whose internal nodes are 
labelled by the parallel composition operator,
each one parametrized by a set of
synchronising actions ($||_A$, with $A \subseteq Act$),
and whose leaves are sequential processes of type $Seq$.
For a specific action $a \in Act$, we write $||_a$ as an abbreviation
to express that $a$ belongs to the set of synchronising actions,
and $||_{\neg a}$ that it does not.
%Let $||_a$ denote an internal node of the process tree of $Sys$ which requires
%synchronisation on $a$.

\begin{definition}
Let $Sys$ be a given SPA system.
\begin{itemize}
\item[(a)]
The set of all sequential processes within $Sys$ is denoted as
$seqproc(Sys)$ 
(i.e.\ the set of all leaves of the process tree of $Sys$).
\item[(b)]
The set of all (sequential or composed) processes within $Sys$ is denoted as
$proc(Sys)$. 
\end{itemize}
\end{definition}
The set $proc(Sys)$ equals the set of all nodes of the
process tree of $Sys$.
Obviously, $seqproc(Sys) \subseteq proc(Sys)$.

Let us denote all actions occurring in the syntactical specification of a
sequential process $P \in seqproc(Sys)$ as $Act(P)$.
We can extend this definition to an arbitrary process $X \in proc(Sys)$
by writing $Act(X) = \bigcup Act(P_i)$,
where the union is over those sequential
processes $P_i$ that are in the subtree of $X$.
For a sequential process $P$, the fact that $a \in Act(P)$ means that
$P$ (considered in isolation) can actually at some point in its dynamic
behaviour perform action $a$.
However, for a process $X \in proc(Sys) \setminus seqproc(Sys)$,
the fact that $a \in Act(X)$ does not necessarily mean that $X$ can actually
perform action $a$.
As an example, think of $X= P \; ||_{a} \; Q$, where $a \in Act(P)$ but
$a \not\in Act(Q)$.
As another example, think of the same $X$ where
$a \in Act(P)$ and
$a \in Act(Q)$ but no combined state is reachable
in which both $P$ and $Q$ can perform action $a$.
Therefore, we define
$Act_{perf}(X) \subseteq Act(X)$ to be those actions that $X$
(considered in isolation) can actually perform.
While $Act(X)$ is a purely syntactical concept, $Act_{perf}(X)$
is a behavioural concept.
%Note that for $P \in seqproc(Sys)$ we have
%$Act_{perf}(P) = Act(P)$, but for a general process $X$ we have
%$Act_{perf}(X) \subseteq Act(X)$.

Given two proceses $X,Y \in proc(Sys)$, we say that $X$ and $Y$ are disjoint if and only if
they do not share any part of the process tree of $Sys$.
Inside the disjoint processes $X$ and/or $Y$, different actions
(from $Act(X)$ and $Act(Y)$) may take place,
among them the specific action $a$, say.
Synchronisation on action $a$ between $X$ and $Y$ is possible if and only if the root of the smallest subtree containing both $X$ and $Y$ is of type $||_a$.
Maximal $||_a$-rooted subtrees are called $a$-scopes, as
formalized in the following definition.
\begin{definition}
\label{def:ascope}
Let $a \in Act$.
An $a$-scope within an SPA system $Sys$ is a subtree
rooted at a node of type $||_a$, provided that on all nodes on the path
from that node to the root of $Sys$ there is no further synchronisation on action $a$
(i.e.\ all nodes on that path, including the root, are of type $||_{\neg a}$).\\
Furthermore, as a special case, if $P \in seqproc(Sys)$ and there is no $a$-synchronisation
on the path from $P$ to the root of the process tree of $Sys$, we say
that $P$ by itself is an $a$-scope.
\end{definition}
For example, in the system shown in Figure~\ref{Scope}, subtrees rooted at $X_1$ and $X_4$ are $a$-scopes, and sequential process $P_3$ is also an $a$-scope.
The only $b$-scope of this system is at the root of the system,
i.e.\ $X_2$.
\begin{figure}[t]
	\centering
	\includegraphics[width=0.8\textwidth]{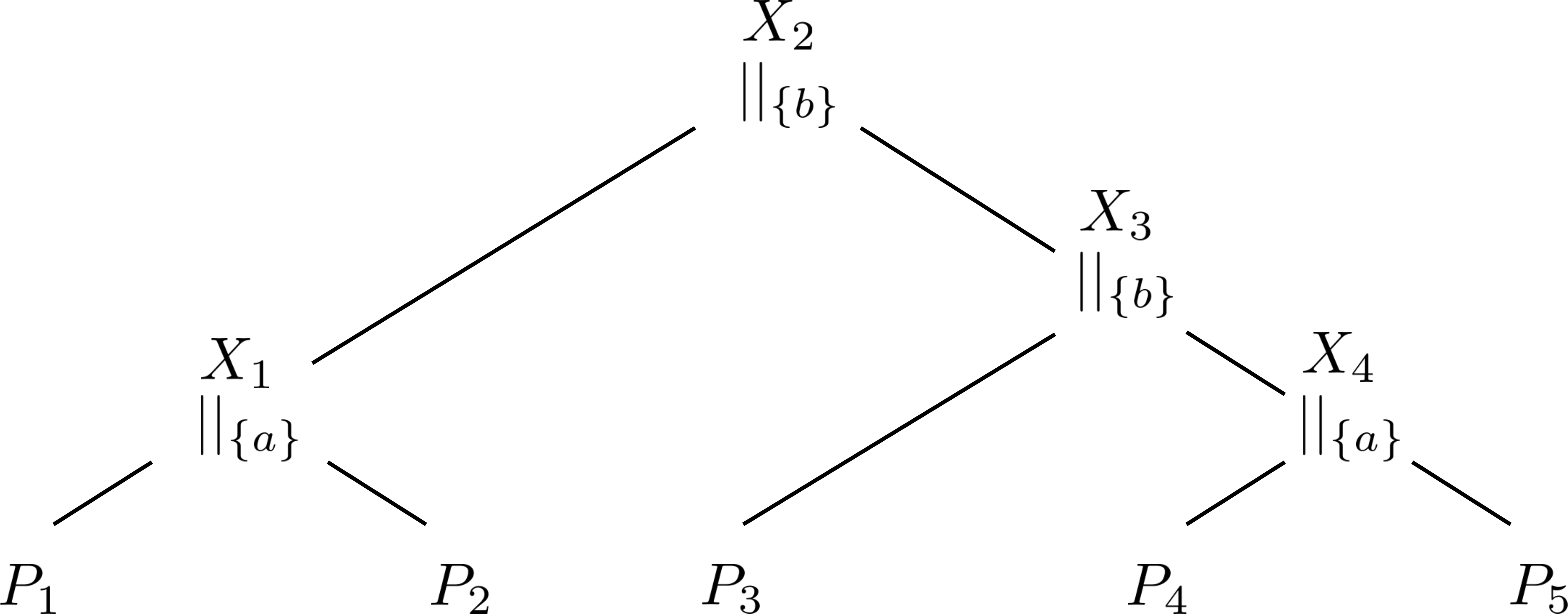}
	\caption{$(P_1 \; ||_{\{a\}} \; P_2) \; ||_{\{b\}} \; (P_3 \; ||_{\{b\}} \; (P_4 \; ||_{\{a\}} \; P_5))$ }
\label{Scope}
\end{figure}

Note that, according to this definition, $a$-scopes are always maximal,
i.e.\ an $a$-scope can never be a proper subset of another $a$-scope.
Clearly, if the root node of $Sys$ requires synchronisation on action $a$, then 
the whole $Sys$ is a single $a$-scope.
Synchronisation via action $a$ is impossible between two distinct $a$-scopes.
But even within a single $a$-scope, not all processes can / need to synchronise
on action $a$.
The following definition answers the question
(from the perspective of a sequential process $P$)
which processes cannot / may / must synchronise with an $a$-transition
in process $P$.
\begin{definition}
\label{def:neighbourhoods}
For $a \in Act$, consider the $a$-transitions within process $P \in seqproc(Sys)$.
Let $X \in proc(Sys)$ be such that $P$ and $X$ are disjoint (i.e.\ that $P$ is not part of $X$).
Let $r$ be the root of the smallest subtree that contains both $P$ and $X$.
\begin{itemize}
\item[(a)]
$X \in N_{cannot}(Sys,P,a)$
iff $r$ is of type $||_{\neg a}$.
%at $r$ there is no synchronisation on $a$.
\item[(b)]
$X \in N_{may}(Sys,P,a)$ iff
$r$ is of type $||_{a}$ but
%at $r$ there is synchronisation on $a$ but
on the path\footnote{``Path'' here means all nodes strictly between $r$ and
the root of $X$.}
from $r$ to $X$ there exists a node of type $||_{\neg a}$.
%without synchronisation on $a$.
\item[(c)]
$X \in N_{must}(Sys,P,a)$ iff
$r$ is of type $||_{a}$ and
%at $r$ there is synchronisation on $a$ and
on the path from $r$ to $X$ all nodes are of type $||_{a}$.
%require synchronisation on $a$.
\end{itemize}
\end{definition}

\begin{remark}
\label{remark:maynot}
Note that for a process $X$ it is possible that
$X \in N_{may}(Sys,P,a)$
or $X \in N_{must}(Sys,P,a)$
even if $a \not\in Act_{perf}(X)$ (or even $a \not\in Act(X)$),
which of course means that $X$ will
never be able to synchronise on action $a$.
This is considered below in the definition of the set $IS_r$ 
(cf.\ Def.~\ref{def:involvedset}).
Related to this observation,
note further that a process $X \in N_{may}(Sys,P,a)$ could actually
be forced to synchronise with $P$ on action $a$
(i.e.\, it could be that $X$ must synchronise with $P$ on $a$,
even though $X \notin N_{must}(Sys,P,a)$).
For example,
if $Sys = P \; ||_a \; (Q \; ||_{\neg a} \; R)$
then $Q \in N_{may}(Sys,P,a)$
and $R \in N_{may}(Sys,P,a)$,
but if $a \notin Act(Q)$
then $P$ always needs $R$ as a synchronisation partner on action $a$.
\end{remark}

\begin{lemma}
\begin{itemize}
\item[(a)]
The neighbourhood $N_{cannot}(Sys,P,a)$ is disjoint from
$N_{may}(Sys,P,a)$ and $N_{must}(Sys,P,a)$.
\item[(b)]
Every $X \in N_{may}(Sys,P,a)$ is a subtree of some $Y \in N_{must}(Sys,P,a)$.
\end{itemize}
\end{lemma}
\begin{proof}
Part (a) follows directly from the definition.
Part(b): For given $X$, one such node $Y$ is the node directly below $r$
on the path from $r$ (as defined in Def.~\ref{def:neighbourhoods}) to $X$.
\end{proof}

\subsection{Moving Set and Participating Set}
\label{subsec:movandpar}
Given a system $Sys$ constructed from $n$ sequential processes $P_1, \dots, P_n$,
its global state is a vector $(s_1, \dots, s_n)$
where $s_i$ is the state of $P_i$.
We follow the convention that the ordering of processes is given by the
in-order (LNR) traversal of the process tree of $Sys$.
A transition $t$ in the flat transition system of $Sys$ is given by
\[
	t = ((s_1, \dots, s_n)  \xrightarrow{a,\lambda_s} (s_1', \dots, s_n'))
\]
where for at least one $k \in \{1, \dots, n \}$ we have $s_k \neq s_k'$
and where the transition rate $rate(t)=\lambda_s$ is a function of the rates of the
transitions of the participating processes.
For such a transition $t$ we introduce the following notation:
\begin{equation*}
\begin{split}
action(t) = a & \qquad rate(t)=\lambda_s \\
 source(t)=(s_1, \ldots, s_n) & \qquad target(t)=(s_1', \ldots, s_n') \\
 source_i(t)=s_i & \qquad target_i(t)=s_i'
\end{split}
\end{equation*}
But which are actually the participating processes in the above transition $t$?
For an $a$-transition $t$ as above, we define the moving set $MS(t)$
as the set of those sequential processes whose state changes, i.e.\
$MS(t)= \{P_k \; | \; s_k \neq s_k' \}$.
The complement of the moving set is called the stable set $SS(t)$, i.e.\
$SS(t) = \{P_1, \dots, P_n  \} \setminus MS(t)$.

Since processes may contain selfloops and since synchronisation on
selfloops is possible
(and often used as a valuable feature to control the context of a transition),
the participating set $PS(t)$ of transition $t$
can also include processes which participate in $t$ in an invisible way
by performing a selfloop.
Therefore $PS(t)$ can be larger than $MS(t)$,
i.e.\ in general we have $MS(t) \subseteq PS(t)$.
Processes in $SS(t)$
which {\em must} synchronise on $a$ with one of the elements of $MS(t)$
must have an $a$-selfloop at their current state and must belong to $PS(t)$.
Furthermore, processes in $SS(t)$ which {\em may} synchronise on $a$ with one of the elements of $MS(t)$
and have an $a$-selfloop at their current state may also belong to $PS(t)$,
provided that they are not in the $N_{cannot}$-neighbourhood
of one of the processes of $MS(t)$.
Altogether we get:
\resizebox{\textwidth}{!}{
\begin{minipage}{\linewidth}
\begin{flalign*}
PS(t) &   =             MS(t)   \; \cup&\\
	  & \phantom{=}     \Big\{P_i \in SS(t) \; |&\\
      & \phantom{=}     \Big(\exists P_j \in MS(t):&\\
	  & \phantom{=}     \big(P_i \in N_{must}(Sys,P_j,a)&\\
      & \phantom{=}     \vee \big(P_i \in N_{may}(Sys,P_j, a) \wedge (\mbox{selfloop} \; 
         s_i \xrightarrow{a, \lambda_i} s_i \; \mbox{exists and is enabled in } source(t)) \big)\big) \Big)&\\
      & \wedge          \big( \not\exists P_j \in MS(t): P_i \in N_{cannot}(Sys,P_j,a) \big) \Big\}
\end{flalign*}
\end{minipage}}
The condition ``selfloop \dots \, is enabled in $source(t)$''
means that the selfloop in $P_i$ can actually take place in the source
state of transition $t$, i.e.\ it
is not blocked by any lacking synchronisation partner(s).
Note that for the case $P_i \in N_{must}(Sys,P_j,a)$
there obviously exists a selfloop in process $P_i$, but this existence
is implicit, so we do not have to write it down.
In Appendix~\ref{AppendixPS}, a procedure for practically calculating $PS(t)$ is given.
%
%{\bf Version 2} (correct, fully formal but complicated):
%\begin{eqnarray*}
%PS(t) & = & MS(t) \cup\\
%	&&
%	\Big\{P_i \in SS(t) \; | \\
%        &&
%        \Big(\exists P_j \in MS(t):\\
%	&&
%	\big(P_i \in N_{must}(Sys,P_j,a) \\
%&& \vee
%	\big(P_i \in N_{may}(Sys,P_j, a)
%	\wedge (\mbox{selfloop} \; s_i \xrightarrow{a, \lambda_i} s_i \;
%\mbox{exists}\\
%&&
%\quad \quad \quad \wedge \;
%(\forall X \in N_{must}^{max}(P_i):
%(X \mbox{ is stable} \Rightarrow (\overrightarrow{s_X} \xrightarrow{a,\lambda_X} \overrightarrow{s_X}) \mbox{ exists}))
%)
%	\big)\big) \Big) \\
%& \wedge & \big(
%        \not\exists P_j \in MS(t): P_i \in N_{cannot}(Sys,P_j,a)
%\big)
%\Big\}
%\end{eqnarray*}
%Here the most complicated subclause is the one covering the case
%$P_i \in N_{may}(Sys,P_j,a)$.
%We need to ensure that all necessary synchronisation partners $X$
%are ready to synchronise with the $a$-selfloop in $P_i$.
%$X$ can be any kind of process, i.e.\ a sequential process or a
%composed process,
%that's why we write $\overrightarrow{s_X}$ instead of $s_X$.
%It is enough to check maximal such synchronisation partners.
%Formally:\\
%$(X \in N_{must}^{max} (P_i)) \Leftrightarrow
%(X \in N_{must} (P_i) \; \wedge  \nexists Y \in N_{must}(P_i):
%X \mbox{ is a subprocess of } Y)
%$
%{\bf Version 3}

Using an example we show why the definition of $PS(t)$
needs to be so complicated, in particular why
being in the $may$ neighbourhood of a moving component and having a selfloop 
is not enough to become a participating component.
For the system shown in Figure~\ref{PSexample}, we wish to
find $PS(t)$ where $t=((1,1,3,1,2) \xrightarrow{a} (2,1,3,1,2))$.
$P_1$ is the only moving component,
and assume that there are $a$-selfloops in state 1 in  $P_2$ and also in state 3 in $P_3$,
but that there are no $a$-selfloops in state 1 of $P_4$ and in state 2 of $P_5$.
$P_2 \in N_{may}(Sys, P_1, a)$ is in $PS(t)$, since its selfloop can take place without hindrance,
whereas
$P_3 \in N_{may}(Sys, P_1, a)$ is not included in $PS(t)$, since its selfloop,
although it exists,
is not enabled in the source state of transition $t$ (it would need $P_4$ or $P_5$ as a synchronisation partner).

\begin{figure}[h!]
	\centering
	\includegraphics[width=1\textwidth]{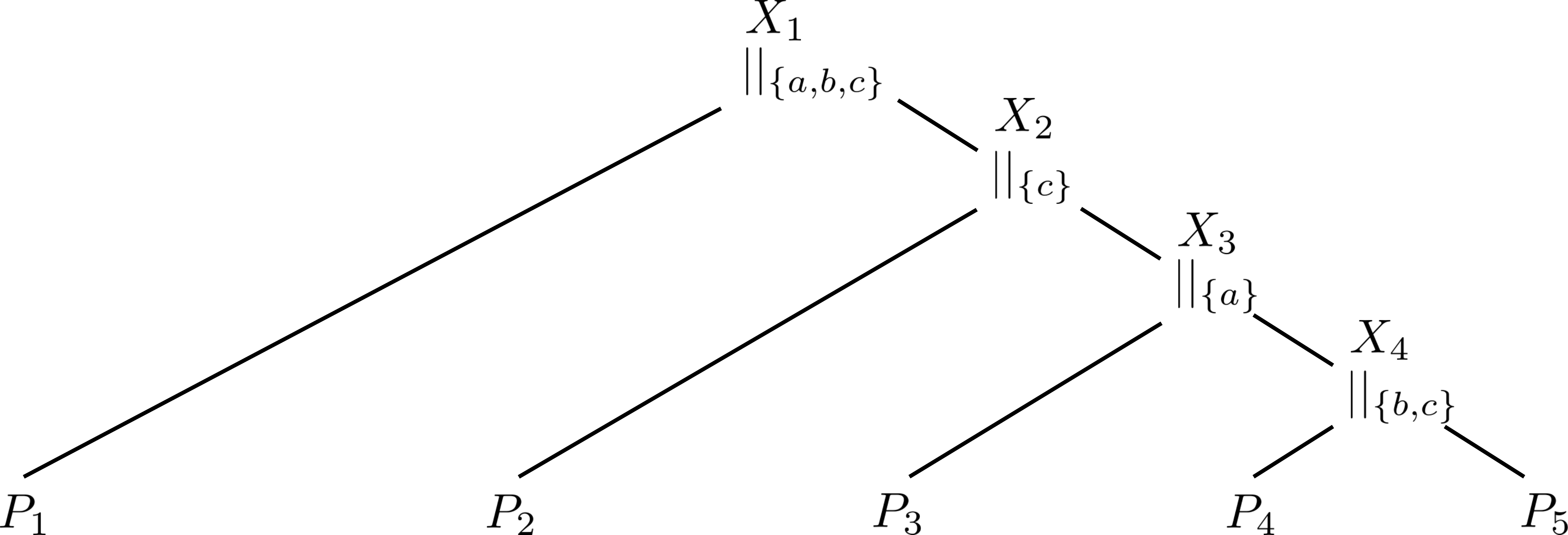}
	\caption{$P_1 \; ||_{\{a,b,c\}}\;(P_2 \; ||_{\{c\}} \; (P_3 \; ||_{\{a\}} \; (P_4 \; ||_{\{b,c\}} \; P_5)))$ }
\label{PSexample}
\end{figure}

\subsection{Involved Set}
In addition to the Participating Set $PS(t)$ of a transition $t$,
we also need to define the Involved Set $IS(t)$ which can be larger than
$PS(t)$, since it also contains those processes which may synchronise
on action $a$ with one of the processes in $PS(t)$ (in another transition $t'$), and so on, inductively.
Formally:

\begin{definition}
	\label{def:involvedset}
For a transition $t$ with $action(t)=a$ we define
\begin{itemize}
\item[(i)]
The Involved Set
\begin{eqnarray*}
IS(t) & = & PS(t) \cup
	\Big\{P_k \in seqproc(Sys) \; | \; \exists P_j \in IS(t):
	\big(P_k \in N_{may}(Sys,P_j,a)
	\big)
	\Big\}
\end{eqnarray*}
\item[(ii)]
The restricted Involved Set
$IS_r(t) = \{P \in IS(t) \; | \; action(t) \in Act(P)  \}$.
\end{itemize}
\end{definition}
So $IS(t)$ represents the convex hull of the $N_{may}$-neighbourhood of one of the participating processes.
That's why after the existential quantor in part (i) we have to write
$P_j \in IS(t)$
instead of only
$P_j \in PS(t)$.
The definition of the restricted Involved Set $IS_r(t)$ is motivated by
the observation in Remark~\ref{remark:maynot}.
The idea is to omit those sequential processes $P$ from $IS(t)$
where $a \not\in Act(P)$, since
they will never actually synchronise on action $a$ with any other process.

In some cases $IS(t) = PS(t)$,
but it can be easily shown by example that $IS(t)$ may be a strict
superset of $PS(t)$.
Consider the system
\[
Sys = ( P_1 \; ||_{ \neg a} \; P_2 ) \; ||_a \; ( P_3 \; ||_a \; P_4)
\]
and the transition
$
t = ((s_1, s_2, s_3, s_4)  \xrightarrow{a,\lambda_s}
(s_1', s_2, s_3', s_4))
$.
The moving set is $MS(t)= \{ P_1, P_3 \}$, and there is obviously a selfloop in $P_4$ of the form
$s_4 \xrightarrow{a,\lambda_4} s_4$, so
the participating set is
$PS(t) = \{ P_1, P_3, P_4  \}$.
However, $P_2$ is also (indirectly) involved since it is possible
that in  some other transition $P_3$ (and $P_4$) will synchronise
on action $a$ with $P_2$.
More concretely:
The transitions
$s_3 \xrightarrow{a, \lambda_3} s_3'$ (in $P_3)$ and
$s_4 \xrightarrow{a, \lambda_4} s_4$ (in $P_4)$
may synchronise with
$s_2' \xrightarrow{a, \lambda_2} s_2''$ (in $P_2)$
(for some states $s_2'$ and $s_2''$ of $P_2$).
Therefore we get $IS(t) = \{ P_1, P_2, P_3, P_4 \}$.
This will be important for our rate lifting algorithm
(Sec.~\ref{sec:LiftingAlgo}),
since if we didn't take the involvement of $P_2$ into account,
we might change some rates in $P_3$ and/or $P_4$ which would
have side effects on other transitions.
This means that we have to set up a system
of equations involving all four processes.

The following lemma establishes the connection between a transition's
Involved Set $IS(t)$ (a behavioural concept)
and the $a$-scope from Def.~\ref{def:ascope}, which latter is a structural concept.

\begin{lemma}
For a transition $t$ of the SPA system $Sys$, with $action(t)=a$,
let $r$ be the root of the smallest tree containing all
processes of $IS(t)$.
\begin{itemize}
\item[(i)]
Then $r$ is a node of type $||_a$.
\item[(ii)]
There is no other node of type $||_a$ ``above'' $r$
(i.e.\ on the path from $r$ to the root of $Sys$).
\item[(iii)]
The Involved Set $IS(t)$ is exactly the set of
all sequential processes in the subtree rooted at $r$.
\item[(iv)]
The Involved Set $IS(t)$ is exactly the set of sequential processes
in the $a$-scope rooted at $r$.
So, in a sense, the Involved Set and the $a$-scope are equal.
\end{itemize}
\end{lemma}
\begin{proof}
	(i)
Assume that $r$ was of type $||_{\neg a}$.
Then no process $P_l \in IS(t)$ in the left subtree of $r$ could 
synchronise (on action $a$)
with any process $P_r \in IS(t)$ in the right subtree of $r$,
which contradicts the fact that the set $IS(t)$ contains
processes in both subtrees of $r$.\\
(ii)
Furthermore, assume that there is another node $r_2$ of type $||_a$ on the
path from $r$ to the root of $Sys$.
Then any $a$ transition in one of the processes of $IS(t)$
would have to synchronise with some process in the other subtree of $r_2$,
which means that the subtree rooted at $r$ does actually not contain
all processes of $IS(t)$, which is a contradiction.\\
(iii)
Assume that there is a sequential process $P_{not}$ in the left subtree
of the tree rooted at $r$ such that $P_{not} \not\in IS(t)$.
We know that there exists a sequential process $P_{r}$ in the right
subtree of the tree rooted at $r$ such that $P_{r} \in IS(t)$.
Then, since according to (i) $r$ is of type $||_a$, either
$P_{not} \in N_{may}(Sys,P_r,a)$ or
$P_{not} \in N_{must}(Sys,P_r,a)$.
%Then any $a$-transition
%in $P_{not}$ would have to synchronise with an $a$-transition in the right
%subtree of the tree rooted at $r$.
But from this it follows that $P_{not}$ would have to be in $IS(t)$,
which is a contradiction.
A symmetric argument holds if we assume that
there is a sequential process $P_{not}$ in the right subtree
of the tree rooted at $r$.
Furthermore, any sequential process $P_{not}$ not in the subtree rooted at $r$
cannot be in $IS(t)$ because according to (ii) there is no
$a$-synchronisation above $r$.\\
(iv)
This is an immediate consequence of (i) - (iii).
\end{proof}

\begin{lemma}
\label{lemma:setmeansequal}
For two transitions $t_1$ and $t_2$ with $action(t_1)=action(t_2)$,
if $IS(t_1) \cap IS(t_2) \neq \emptyset$
then
$IS(t_1) = IS(t_2)$.
\end{lemma}
\begin{proof}
This follows directly from the closure property of the $IS$ definition.
\end{proof}

\section{``Parallel'' Transitions and Relevant Selfloop Combinations}
\subsection{Multi-Transition System}
It is well known that the semantic model underlying an SPA specification
is actually a {\em multi}-transition system \cite{HillstonBook,Goetz94}.
This is usually flattened to an ordinary transition system by
adding up the rates of ``parallel'' transitions, i.e.\ transitions which have
the same source state,
the same target state
and the same action label.
Thus a transition within $Sys$ may be the aggregation of more
than one transition.
As an example, consider the system
$
Sys = P \; ||_a \; (Q \; ||_{\neg a} \; R)
$
and the transition
$
t = ((s_1, s_2, s_3)  \xrightarrow{a,\lambda_s} (s_1', s_2, s_3))
$.
The moving set is $MS(t)= \{ P \}$, but the participating set
must be larger.
Assume that $Q$ has a selfloop $s_2 \xrightarrow{a,\lambda_2} s_2$ and
that $R$ has a selfloop $s_3 \xrightarrow{a,\lambda_3} s_3$.
Since $Q$ and $R$ do not synchronise on $a$, only one of those two selfloops
synchronises with $s_1 \xrightarrow{a,\lambda_1} s_1'$ at a time, but
both selfloops may synchronise with the $a$-transition in $P$.
This yields the two ``parallel'' transitions
\begin{equation*}
((s_1, s_2, s_3)  \xrightarrow{a,\lambda_{12}}  (s_1', s_2, s_3)) \;\;\;\mbox{and}\;\;\;
((s_1, s_2, s_3)  \xrightarrow{a,\lambda_{13}} 
 (s_1', s_2, s_3))
\end{equation*}
(where $\lambda_{12}$ is a function of $\lambda_1$ and $\lambda_2$,
and likewise for $\lambda_{13}$)
which are aggregated to the single transition
$
((s_1, s_2, s_3)  \xrightarrow{a,\lambda_{12}+\lambda_{13}}  (s_1', s_2, s_3))
$,
so $\lambda_s = \lambda_{12} + \lambda_{13}$.
%{\color{red}
As an anticipation of Eq.~\ref{eqn:bigeqn} in Sec.~\ref{sec:LiftingAlgo}, let us mention that in this situation our rate lifting algorithm
would create the equation
$
x^{(P)}_{s_1 s_1'} x^{(Q)}_{s_2 s_2} + x^{(P)}_{s_1 s_1'} x^{(R)}_{s_3 s_3} = \lambda_s \cdot f
$.
%}

\subsection{Calculating Relevant Selfloop Combinations}
\label{subsec:rslc}

In the simple (and most common) case that none of the
sequential processes in the SPA specification of $Sys$
has any selfloops (and also no ``parallel'' transitions), we know that
any transition of the flat transition system has only one
single semantic derivation.
In consequence, for the considered flat transition $t$ it then holds
that $PS(t) = MS(t)$.
However, as discussed above, in the general case
the flat transition system underlying a
compositional SPA specification is actually a multi-transition system 
which gets flattened to an ordinary transition system by
amalgamating ``parallel'' transitions.
%During the amalgamation, the rates of the parallel transitions are
%simply summed up.
%
In order to cover this general case, in the lifting algorithm
(see Sec.~\ref{sec:LiftingAlgo})
we have to do the opposite:
Instead of amalgamation,
we need to deconstruct a flat transition into its constituents.
I.e., given a flat transition
(which is possibly amalgamated from parallel transitions),
we need to find out the
contributing transitions, in order to be able to construct the
correct equation in line~\ref{line:createeq3} of the algorithm
(Eq.~\ref{eqn:bigeqn} in Sec.~\ref{sec:LiftingAlgo}).

Consider the flat transition
$t := ((s_1, \dots, s_n) \xlongrightarrow{c, \gamma \cdot f} (s_1', \dots, s_n'))$.
We can determine its (non-empty) moving set $MS(t)$
and its participating set $PS(t)$,
where we know that $MS(t) \subseteq PS(t)$.
We are particularly interested in the processes from the set
$(PS(t) \cap SS(t)) \setminus \bigcup_{P \in MS(t)} N_{must}(Sys,P,c)$,
%$(PS(t) \setminus MS(t)) \setminus \bigcup_{P \in MS(t)} N_{must}(Sys,P,c)$,
since these are exactly the processes that may (but not must)
contribute to transition $t$.
Certain combinations of these processes
(which have selfloops, otherwise they wouldn't be in $PS(t)$)
contribute to transition $t$.
We call these combinations ``relevant selfloop combinations (rslc)''.
Note that there are also selfloops in
$\bigcup_{P \in MS(t)} N_{must}(Sys,P,c)$,
but they are not part of rslc.

It remains to calculate rslc for transition $t$.
For this purpuse, we define a function $rslc(t)$ which returns
a set of sets of sequential processes, i.e.\ each such set
describes a relevant selfloop combination.
In the process tree of $Sys$,
let $r$ be the root node of the smallest subtree containing $PS(t)$.
We know that $r$ is either an inner node of type $||_c$
or a leaf (if $r$ were an inner node of type $||_{\neg c}$,
the participating set $PS(t)$ couldn't span both subtrees of $r$).
Calling the recursive algorithm in Appendix~\ref{AppendixRSLC}
by the top-level call RSLC($t,r$)
delivers all the relevant selfloop combinations\footnote{
Note that $rslc(t)$ as called in the lifting algorithm has one argument
(a transition)
but the recursive function $RSLC(t,n)$ (see Appendix~\ref{AppendixRSLC}) has two arguments (a transition and
a node of the process tree).}.

\begin{figure}[t!]
	\centering
	\includegraphics[width=0.8\textwidth]{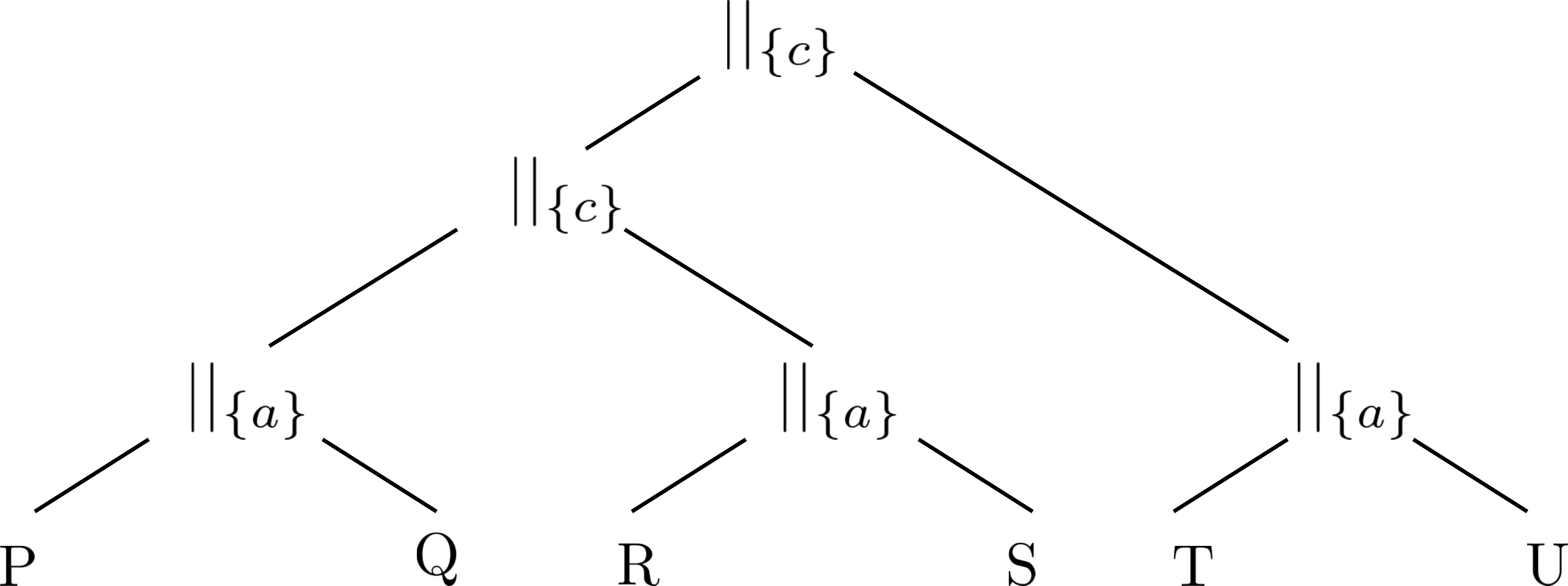}
	\caption{$
(
(P \; ||_{\neg c} \; Q)
\; ||_{c} \;
(R  \; ||_{\neg c} \; S)
)
\; ||_{c} \;
(T \; ||_{\neg c} \; U)
$ }
\label{rslc}
\end{figure}
Example: Consider the SPA specification
$
Sys = 
(
(P \; ||_{\neg c} \; Q)
\; ||_{c} \;
(R  \; ||_{\neg c} \; S)
)
\; ||_{c} \;
(T \; ||_{\neg c} \; U)
$
whose process tree is shown in Figure~\ref{rslc}, and the transition
$
t := ((s_P, s_Q, s_R, s_S, s_T, s_U) \xlongrightarrow{c, \gamma \cdot f}
(s_P', s_Q, s_R, s_S, s_T, s_U)
)
$.
Obviously, the moving set is $MS(t) = \{P \}$,
and if we assume that there are $c$-selfloops in states
$s_R, s_S, s_T$ and $s_U$ (in all of them!),
the participating set is $PS(t) = \{P, R, S, T, U \}$.
So transition $t$ can be realised as any combination of a selfloop
in $R$ or $S$ with a selfloop in $T$ or $U$,
thus the algorithm will find the set of relevant selfloop combinations
$\{ \{R,T \}, \{R, U \}, \{S,T \}, \{S,U \} \}$.\\
%{\color{red}
Anticipating once again Eq.~\ref{eqn:bigeqn} from Sec.~\ref{sec:LiftingAlgo},
this set of relevant selfloop combinations would lead to
the desired equation
\[
x^{(P)}_{s_P s_P'}  x^{(R)}_{s_R s_R}  x^{(T)}_{s_T s_T}
+
x^{(P)}_{s_P s_P'}  x^{(R)}_{s_R s_R}  x^{(U)}_{s_U s_U}
+
x^{(P)}_{s_P s_P'}  x^{(S)}_{s_S s_S}  x^{(T)}_{s_T s_T}
+
x^{(P)}_{s_P s_P'}  x^{(S)}_{s_S s_S}  x^{(U)}_{s_U s_U}
= \gamma \cdot f
\]
Alternatively, if we assumed that the participating set was smaller,
say $PS(t) = \{P, R, S, T \}$ (i.e.\ if there were no $c$-selfloop at $s_U$),
then the algorithm would find a smaller set of relevant selfloop combinations,
namely
$\{ \{R,T \}, \{S,T \} \}$,
leading to the simpler equation
$
x^{(P)}_{s_P s_P'}  x^{(R)}_{s_R s_R}  x^{(T)}_{s_T s_T}
+
x^{(P)}_{s_P s_P'}  x^{(S)}_{s_S s_S}  x^{(T)}_{s_T s_T}
= \gamma \cdot f
$
%}

\section{Lifting Algorithm}
\label{sec:LiftingAlgo}

Our new lifting algorithm processes the transitions whose rates are to be modified in a one by one fashion.
It is, however, not strictly one by one, since in many situations a whole set of ``related'' transitions is taken into account together with the currently processed transition.
The algorithm consists of four parts named A, B, C and D
%It starts either with part A or part B.
%From part A it may need to move on to part D, whereas from part B it may need to move on to part C, and from
%there also to part D
as shown in Figure~\ref{fig:algooverview}.
\begin{figure}[t]
\centering
	\includegraphics[width=0.85\textwidth]{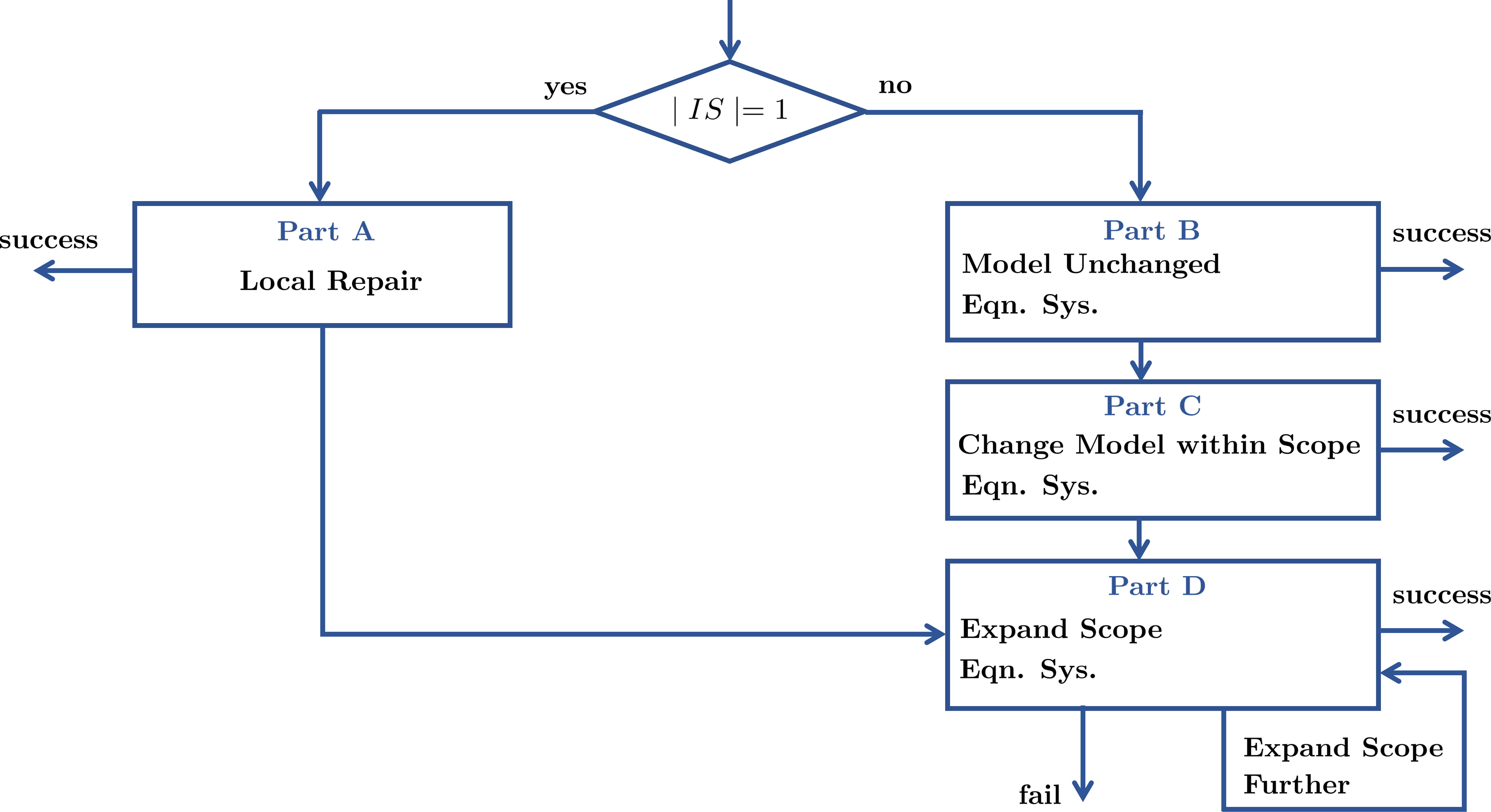}
	\caption{Overview of the algorithm}
  \label{fig:algooverview}
\end{figure}
In part A, for a transition whose involved set consists of only one single sequential process, the algorithm first tries to change its rate by local repair, which means changing the rate locally in exactly this sequential process.
Local repair will fail, however, if two flat transitions which both originate from the same local transition have different modification factors.
In part B, which is the starting point for transitions whose involved set contains at least two processes, the algorithm creates a system of nonlinear equations and tries to solve it. This system of equations covers all transitions with the same action label and the same involved set, i.e.\ all these transitions are dealt with simultaneously in one system of equations.
The basic idea behind the system of equations is to consider all involved local rates as variables whose values are to be determined.
Part C, entered upon failure of Part B, is the first part where the system specification is modified by augmenting some synchronisation sets and inserting selfloops, all within the current $c$-scope.
These modifications are done in such a way that the global transition system is not changed.
Again, like in part B, the algorithm creates a set of nonlinear equations (but now the system of equations is larger since the model has been modified) and tries to solve it.
If the previous steps have failed, Part D tries to expand the scope, by modifiying the system in a larger scope than the current involved set.
This means that the involved set is artificially augmented by adding action $c$ to the synchronisation set at a higher node.
Again, a similar but even larger system of equations of the same type is constructed.
However, even this system of equations may not have a solution, in which case the desired rate lifting has turned out to be impossible.

\subsection{Spurious Transitions}
As we have seen, in certain situations the rate lifting algorithm
needs to change the synchronisation structure of the given system, i.e.\
it will change an inner node of type $||_{\neg c}$ to a node of type  $||_{c}$.
Clearly, this needs to be done with great care, since such a step will -- in general --
change the behaviour of the system.
Therefore the algorithm, before adding action $c$ to a synchronisation
set, has to ensure that no spurious transitions will be generated.
Spurious transitions (sp.\ tr.) are extra, superfluous transitions not present
in the original system, and therefore incorrect.
Furthermore, after action $c$ has been added to a synchronisation set,
the algorithm also has to ensure that all transitions in the original
system are still possible (it could easily be that a previously
existing transition now lacks a synchronisation partner in the newly
synchronised system).
For that purpose, the algorithm inserts selfloops into the sequential
components, wherever necessary.
There are actually two types of spurious transitions:
\begin{itemize}
\item[(A)]
Superfluous transitions which appear when two previously
$c$-non-synchronised components become synchronised over action $c$.
\item[(B)]
Superfluous transitions which appear when a new $c$-selfloop is inserted
into a sequential process which is in the $c$-must- or $c$-may-neighbourhood
of another process.
\end{itemize}
Overall, the algorithm guarantees that even though the synchronisation
structure of the system may be altered and artificial selfloops are inserted,
the set of reachable states
and the set of transitions remain the same.

\subsection{Informal Description of the Algorithm}

Appendix~\ref{AppendixBigAlgo} shows our new rate lifting algorithm.
The arguments of the algorithm are the SPA system $Sys$ and
its flat transition system $T$ (a set of transitions),
the set of transitions whose rate is to be modified $T_{mod} \subseteq T$
as well as a function $factor$ that returns, for each transition in $T_{mod}$,
its modification factor\footnote{Thus,
this presentation of the algorithm addresses
model repair rather than rate reengineering (cf.\ Sec.~\ref{sec:Intro}).}.
For transitions not in $T_{mod}$,
the modification factor is supposed to be 1.

In each iteration of the outer while-loop, the algorithm
picks one of the remaining transitions from $T_{mod}$, 
called $\hat{t}$ with action label called $c$, and processes it (possibly together with
other transitions that have the same action label).

\noindent {\bf (Part A) Local repair:}
If the involved set of the currently processed transition $\hat{t}$
consists of only one single process, the algorithm tries to 
adjust the rate of exactly one transition in that process.
This is done in lines \ref{line:localrepairbegin}--\ref{line:localrepairend} of the algorithm.
However, this will only work if all transitions where
this process makes the same move have the same, common modification factor.
That set of transitions is denoted as
$T_{c,\hat{s}_{i_1},\hat{s}_{i_1}'}$ in the algorithm,
and the common modification factor is denoted as $f_{com}$.

\noindent {\bf (Part B) System of equations for $T_c$:}
If the involved set of the currently processed transition $\hat{t}$
consists of two or more processes, all transitions with the same
involved set and the same action as $\hat{t}$ are processed together
(lines \ref{line:originalTcBegin}--\ref{line:originalTcEnd}).
In the algorithm, this set of transitions is denoted $T_c$.
For every transition $t \in T_c$, the algorithm determines
its participating set $PS(t)$, calculates the
relevant selfloop combinations (rslc) and from this information
creates a nonlinear equation
\begin{equation}
\label{eqn:bigeqn}
\sum_{C \in rslc(t)}
\prod_{P \in MS(t)} x^{(P)}_{s_P s_P'}
\prod_{
{
Q \in PS(t) \setminus MS(t) \atop \wedge \; \exists P \in MS(t): \; Q \in N_{must}(Sys,P,c)
}
}
\!\!\!\!\!\!\!\! x^{(Q)}_{s_Q s_Q} \;\;\;\;
\prod_{R \in C} x^{(R)}_{s_R s_R}
= \gamma \cdot f
\end{equation}
where the $x$'s are the unknown rates of the participating processes
(some of which are rates of selfloops, if such exist in the system).
The superscript of variable $x^{(\cdot)}_{\cdot\cdot}$ identifies the sequential process, and the subscript denotes the source/target pair of states.
The equation reflects the fact that the rates of all synchronising
processes are multiplied\footnote{Multiplication of rates is a de facto standard for Markovian SPAs, as implemented, for example, by the tools PRISM and STORM.
If the rate resulting from the synchronisation of two or more processes were
defined other than the product of the participating rates,
the equation would have to be changed accordingly,
but apart from this change, the lifting algorithm would still work in the same way.},
and that the total rate is obtained as the
sum over all possible relevant selfloop combinations.
Afterwards, this system of equations is solved, and if a solution exists,
all $c$-transitions in the current $c$-scope have been successfully
dealt with.
We would like to point out that,
if for some transition $t$ the participating set $PS(t)$ is
equal to its moving set $MS(t)$, then the resulting equation
has a much simpler form
\[
x^{(P_1)}_{s_1 s_1'} \cdot x^{(P_2)}_{s_2 s_2'}
\cdot \dots \cdot x^{(P_k)}_{s_k s_k'} = \gamma \cdot f
\]
(assuming that $|MS(t)|=k$),
since in this case, there are no selfloops involved, and therefore also
no combinations of selfloops to be considered.

\noindent {\bf (Part C) Expanding the context of $c$-transitions by synchronising with
more processes and inserting artificial selfloops within the current
$c$-scope:}
If the system of equations constructed in part (B)
for the set $T_c$ had no solution,
it is the strategy of the algorithm to involve more processes
(for the moment only from the current $c$-scope),
since this opens up more opportunity for controlling the context of these
$c$-transitions, and thereby controlling their rates.
In this part of the algorithm
(lines \ref{line:extendedTcBegin}--\ref{line:extendedTcEnd}),
$action(\hat{t}) = c$ is added to the synchronisation set at each node of
type $||_{\neg c}$ in the current $c$-scope,
except where this would lead to spurious transitions (of type A or type B).
These tasks of the algorithm are outsourced to function TRYSYNC (called in line \ref{line:trysync1}).
In TRYSYNC (the details of which are elaborated on in Appendix~\ref{AppendixTRYSYNC}), checking for spurious transitions of type A is done by checking all source
states of transitions in the current $T_c$, making sure that there are no
concurrently enabled $c$-transitions in newly synchronised subprocesses.
%(line~\ref{line:notConcEnabled} of function TRYSYNC).
%(it is enough to avoid spurious transitions emanating from reachable states,
%since spurious transitions from unreachable states do not hurt).
After adding action $c$ to some synchronisation sets,
we also have to make sure that all transitions originally in $T_c$ can still
occur, i.e.\ that they have not been disabled by the new synchronisations.
This is also done
%in line~\ref{line:addselfloops} of
in function TRYSYNC,
by inserting the necessary selfloops in those processes which
are now newly synchronising on action $c$, provided that those
new selfloops do not lead to the existence of spurious transitions
(of type B).
The steps just described guarantee that the modified system $Sys'$ has
exactly the same set of transitions as the original system $Sys$
(qualitatively),
but it remains to find the correct rates of all $c$-transitions
in the involved processes.
For this purpose, a similar (but larger) system of equations as in part (B)
is set up and solved.

\noindent {\bf (Part D) Expanding the Involved Set by moving
the current root upwards:}
It is possible that the systems of equations constructed in part (B)
and thereafter in part (C) both have no solution.
In this case, the algorithm seeks to expand the current $c$-scope by moving
its root up by one level
(unless the root of the overal system has already been reached).
Again, it needs to be ensured that no spurious transitions would
be created from this step.
This is done in lines \ref{line:moveupBegin}--\ref{line:moveupEnd}
of the algorithm, again with the help of function TRYSYNC.

%%%%%%%%%%%%%%%%%%%%%%%%%%%%%%%%%%%%%
\section{Experimental Result: Cyclic Server Polling System}
This section considers -- as a case study -- the Cyclic Server Polling System from the PRISM CTMC benchmarks, originally
described in \cite{IT90} as a GSPNs.
It is a system where a single server 
polls $N$ stations and provides service for them in cyclic order.
The SPA representation of this system is:
\begin{equation*}
Sys=Server \; ||_{\Sigma_s} \; (Station_1 \; || \; Station_2 \; ||\ldots|| \; Station_N)
\end{equation*}
where 
$\Sigma_s=\{loop_{ia},loop_{ib},serve_i \; | \; i=1 \dots N\}$.
%A more detailed description of the model is available in~\cite{IT90}. 
Assume that for each $loop_{1a}$-transition $t$ in the combined flat model, a modification factor $f(t) \neq 1$ is given.
Using our new lifting algorithm (Appendix~\ref{AppendixBigAlgo}), we lift this model 
repair information to the components.
The modification factors $f$ are chosen in such a way that local repair (Part A) and also
Part B of the algorithm will not find a solution.
In part C of the algorithm, it turns out that action $loop_{1a}$ can be added to all  
$||_{\neg loop_{1a}}$-nodes of the process tree, since it does not cause spurious transitions. 
Consequently, $loop_{1a}$-selfloops are added to all the states of the 
components $Station_2,\ldots,Station_{N}$,
leading to the modified SPA model
\begin{equation*}
Sys'=Server||_{\Sigma_s}(Station_1||_{loop_{1a}} Station'_2||_{loop_{1a}} \dots ||_{loop_{1a}}Station'_N)
\end{equation*}
where the stations with added selfloops are shown by $Station'_i$.
With the chosen modification factors, a solution can be found in Part C of the algorithm.
%Since in the SPA model $Sys$, there is only one  $loop_{1a}$-scope, which is the root of the process tree ($||_{\Sigma_s}$), expanding the context outside
%of the current $looop_{1a}$-scope is not an option (Part D of the algorithm is impossible).
%
This example is a scalable model where the state space increases with the number of stations $N$.
Note that the model contains symmetries, but the considered rate lifting problem is not symmetric, since only 
$Station_1$ and the $Server$ participate in the $loop_{1a}$-transitions.

\begin{table}[t]
\caption{Model statistics of the combined model for different numbers of stations}
\label{PollTable}
%{\scriptsize
\begin{tabular}{| l | c |c|c| c|c|c|}
  \hline 
  \backslashbox{ \hspace{1.25cm} }{N}          & 6    &7       &8      & 9     &10        &11    \\
  \hline
    Total number of states     				   &576   & 1344   &3072   &6912   &15360     &33792 \\		
  \hline
   Total number of transitions  			   &2208  & 5824   &14848  &36864  &89600     &214016\\                           
  \hline
  Number of $loop_{1a}$-transitions 		   &32    & 64     &128    &256    &512       &1024  \\                
  \hline 
\end{tabular}
%}
\end{table}
Table~\ref{PollTable} shows the model 
statistics for different numbers of stations.
The last row of the table (number of $loop_{1a}$-transitions) equals the number of equations, each of the $2^{N-1}$ equations containing the product of $N+1$ unknown variables.
The whole system of equations 
has $(N-1)*2+2$ variables stemming from $(N-1)*2$ newly added $loop_{1a}$-selfloops 
plus two original $loop_{1a}$-transitions (in the $Server$ and $Station_1$).
Figure~\ref{resPoll} shows the required times to run our rate lifting algorithm
(implemented as a
proof-of-concept prototype in Matlab \cite{Matlab}) and to solve the system of 
equations (done by Wolfram Mathematica~\cite{Mathematica})
 for different values of $N$ \footnote{Executed on a standard laptop with Intel Core i7-8650U 
CPU@ 1.90GHz-2.11GHz }.
For large $N$, the time for equation solving by far dominates the runtime of our algorithm
 (by a factor of 2.61 for $N=11$).
%For  different numbers of stations, the modification 
%factors are chosen so that the system of equations has a solution.
As shown in the figure, the runtimes grow exponentially, which is not surprising since 
the number of equations increases exponentially.
\begin{figure}[t!]
	\centering
	\includegraphics[width=0.7\textwidth]{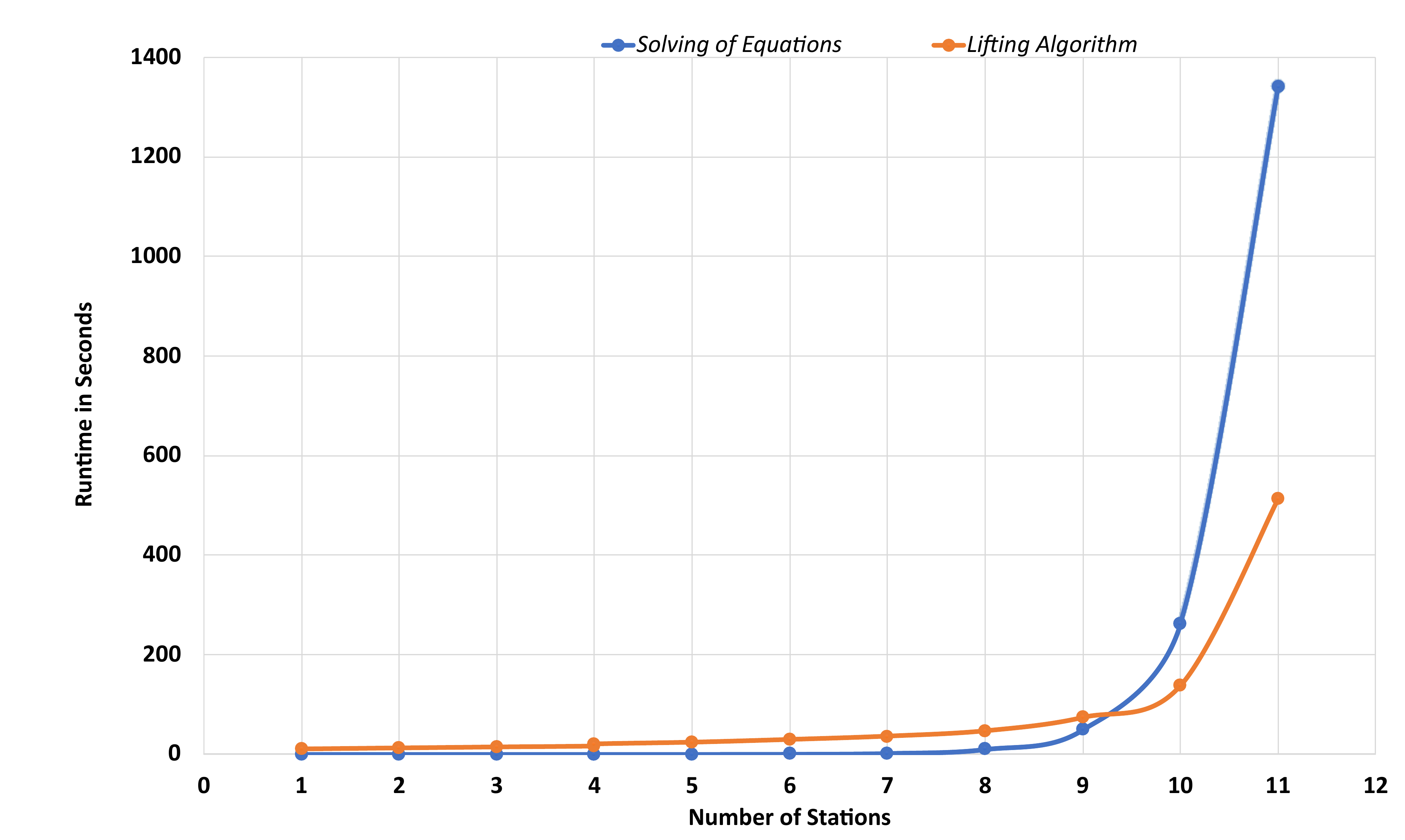}
	\caption{Runtime comparison for different number of stations}
  \label{resPoll}
\end{figure}

\section{Correctness and Optimality of the Algorithm}
\label{AppendixCorrectOptimal}

{\bf Correctness:}
It must be guaranteed that a solution found by the algorithm is correct, which means that
the modified system (with the calculated rates, possibly modified synchronisation sets and added selfloops)
possesses the same transition system, just with the transition rates modified as desired.
Once a solution has been found by the algorithm, it is easy to check its correctness
by simply constructing the flat transition system for the modified system and comparing it to the desired transition system.
However, one can also constructively argue for the correctness of the algorithm:
In Part A, if the condition in line \ref{line:commonfactorexists} is fulfilled, only the transition rate of a single transition
in one of the sequential processes is changed,
resulting in the change of a well-defined set of global flat transitions, all having the same modification factor,
which is the intended result.
Parts B, C and D each work by setting up and solving a nonlinear system of equations
relating to the original (Part B) resp.\ carefully modified SPA system (Parts C and D).
Each of these systems of equations precisely reflects the synchronisation
of sequential SPA processes within a certain scope,
taking into account all transitions with action label $c$ that take place in that scope, 
and making sure that the resulting rates of those transitions are all as desired
(thereby considering all relevant selfloop combinations).
All procedures in Parts B, C and D affect only a certain scope of the overall SPA system,
so it is enough to ensure correctness for such a local context.
If in Parts C and D the model is adjusted (by augmenting synchronisation sets and inserting artificial selfloops),
care is taken that this will not affect the structure of the low-level transition system.
Thus, since each individual step of the algorithm is correct, we can conclude by induction that the total effect of multiple steps is also correct.\\
{\bf Optimality:}
If the lifting algorithm doesn't find a solution, is it really guaranteed that there doesn't exist one?
We do not provide a formal proof of optimality, but we briefly give the basic line of argument:
If we start with part A of the algorithm and if local repair fails, this happens because the same local transition
(involving only a single sequential process) should be executed in different contexts with different modification factors (i.e.\ different rates), which is of course not possible
in the unmodified system.
In order to solve this problem, some ``controlling'' context needs to be added.
For this purpose, we synchronise the process with its neighbouring processes (where selfloops are added at specific states) in a subtree of a certain height, which leads to a set of equations (of the form of Eq.~\ref{eqn:bigeqn}) in part D of the algorithm.
We keep expanding the context until either a solution has been found or the root of the system has been reached,
which means that the algorithm uses its full
potential.
Alternatively, if we start with part B (because the involved set of the currently processed transition is already larger than one),  we first search for a solution in the ``local'' context, i.e.\ in the current involved set, which is a subtree of the system.
First we try to leave the model unchanged, which also leads to a system of equations.
If it turns out that this system of equations is inconsistent (thus not having a solution), we need to include more degrees of freedom into the equations.
This is first done within the current scope (by synchronising with as many processes as possible, albeit all from within this same scope) in part C.
If this also fails, i.e.\ if the thus extended system is also inconsistent, even more degrees of freedom can be added by
expanding the current scope, leading us again to part D of the algorithm.
In total, the algorithm uses all possible degrees of freedom, since at every step it involves all processes, except those whose involvement would cause damage (in the sense that spurious transitions would occur).
Therefore, since the algorithm uses all possible degrees of freedom, it is optimal.

\section{Conclusion}
In this paper, we have studied some novel structural concepts of Markovian SPA, which enabled us to formulate
an algorithm for the lifting of rate information from the flat low-level transition system of a general SPA model to its components.
The algorithm works for SPA specifications with an arbitrary structure and any number of components.
We have also presented a small case study that illustrates the practical use of the algorithm and remarked on the correctness and optimality of the algorithm.
As future work, we are planning to develop improved implementation strategies for the algorithm.
Another important point for future work is to characterise a priori the set of problem instances for which a solution to the rate lifting problem exists.

%%%%%%%%%%%%%%%%%%%%%%%%%%%%%%%%%%%%%%%%%%%%%%%%%%%%%%%%%%%%%%%%%%
  \bibliographystyle{plain}
  \bibliography{Literatur}
%%%%%%%%%%%%%%%%%%%%%%%%%%%%%%%%%%%%%%%%%%%%%%%%%%%%%%%%%%%%%%%%%%
\newpage
\begin{subappendices}
\renewcommand{\thesection}{\Alph{section}}

%%%%%%%%%%%%%%%%%%%%%%%
\newpage
\section{Rate Lifting Algorithm for an SPA System with $n$ Sequential Components}
\label{AppendixBigAlgo}

\begin{algorithmic}[1]
\small
\STATE {\bf Algorithm} RateLifting ($Sys,T,T_{mod},factor)$
\STATE // $T$ is the flat Markovian transition system of SPA system $Sys$,
\STATE // consisting of sequential processes $P_1, \dots, P_n$ as leaves of a process tree
\STATE // with synchronisation sets $A_i$
\STATE // The algorithm lifts the repair information given in the form
of
\STATE // rate modification factors $factor(t)$ for transitions $t \in T_{mod} \subseteq T$
\STATE // to the high-level components of $Sys$,
if possible.
\STATE // The repaired system is returned as $P_1', \dots, P_n'$
\STATE // and possibly modified synchronisation sets $A_i'$
\STATE
\STATE $P_1':=P_1, \dots, P_n':=P_n$, $\forall i: A_i' := A_i$ // initialisation
\STATE
\WHILE{$T_{mod} \neq \emptyset$}
\STATE choose $\hat{t} := ((\hat{s}_1,\dots,\hat{s}_n) \xlongrightarrow{c, \hat{\gamma} \cdot \hat{f}} (\hat{s}_1',\dots, \hat{s}_n'))$ from $T_{mod}$
\STATE // $\hat{t}$ is the transition processed during one iteration of the outer while-loop
\STATE $found := false$
\STATE // indicates that no solution found yet while processing the current $\hat{t}$
\STATE determine $IS(\hat{t}) := \{P_{i_1}, \dots, P_{i_m}  \}$ // the Involved Set $IS$ is also a $c$-scope
\label{line:mmm}
\STATE
\IF{$|IS(\hat{t})|=1$}
\label{line:localrepairbegin}
\STATE // {\bf Algorithm Part A:}
\STATE // try local repair in $P_{i_1}$ by considering all ``parallel transitions''
\STATE // (same approach as in algorithm for $n=2$ in lines 37-47 of \cite{Soltanieh:2020})
\STATE $T_{c,\hat{s}_{i_1},\hat{s}_{i_1}'} := \{t \in T \; | \; action(t) = c \wedge source_{i_1}(t) = \hat{s}_{i_1} \wedge target_{i_1}(t) = \hat{s}_{i_1}' \}$
\IF{$\exists f_{com} \in \Reals: \forall t \in T_{c,\hat{s}_{i_1},\hat{s}_{i_1}'}: factor(t) = f_{com}$} \label{line:commonfactor}
\label{line:commonfactorexists}
\STATE // there exists a common factor $f_{com}$ for all transitions in $T_{c,\hat{s}_{i_1},\hat{s}_{i_1}'}$
\STATE in $P_{i_1}'$ set $\hat{s}_{i_1} \xlongrightarrow{c, \gamma_{i_1} \cdot f_{com}} \hat{s}_{i_1}'$ (where $\gamma_{i_1}$ is the current rate in $P_{i_1}'$)
\STATE $T_{mod} := T_{mod} \setminus T_{c,\hat{s}_{i_1},\hat{s}_{i_1}'}$
\FOR{each $t \in T_{c,\hat{s}_{i_1},\hat{s}_{i_1}'}$}
  \STATE $factor(t) := 1$
  %\STATE {\color{red} Update $T_{mod}$}
  \STATE // the modification factor of the fixed transitions is changed to 1,
  \STATE // which is important in case they are considered again
  \STATE // when dealing with another $c$-transition from $T_{mod}$ later
  \ENDFOR
\STATE $found=true$
\ENDIF
\label{line:localrepairend}
\IF{$found=false$}
\STATE // local repair was not successful
\STATE // therefore, since $|IS(\hat{t})|=1$, the algorithm has to move up in the tree
\STATE $curr\_root := P_{i_1}$
\STATE // $curr\_root$ denotes the root of the subtree that is currently
\STATE // considered as the context for transition $\hat{t}$
\ENDIF
%\STATE
\newpage
\ELSE
\STATE // {\bf Algorithm Part B:}
\STATE // it holds that $|IS(\hat{t})| > 1$
\STATE $T_c := \{t \in T \; | \; action(t) = c \; \wedge \; IS(t)= IS(\hat{t}) \}$
\label{line:originalTcBegin}
\STATE // all $c$-transitions in current $c$-scope are considered together
\STATE $r :=$ root of current $c$-scope // needed in lines~\ref{line:createeq1}
and \ref{line:rslc2}
\FOR{each $t := ((s_1, \dots, s_n) \xlongrightarrow{c, \gamma \cdot f} (s_1', \dots, s_n')) \in T_c$}
\label{line:aaa}
\STATE determine $PS(t) := \{ P_{p_1}, \dots, P_{p_k} \} \subseteq IS(\hat{t})$
\label{line:ccc}
\label{line:createeq0}
\STATE find the set of relevant selfloop combinations rslc$(t) :=$ RSLC$(t,r)$
\label{line:createeq1}
\STATE // for rslc see Appendix~\ref{AppendixRSLC} and Sec.~\ref{subsec:rslc}
\STATE create an equation
\label{line:createeq2}
%\STATE $x^{(p_1)}_{s_{p_1}s_{p_1}'} \cdot x^{(p_2)}_{s_{p_2}s_{p_2}'} \cdot \ldots \cdot x^{(p_k)}_{s_{p_k}s_{p_k}'} = \gamma \cdot f$
\STATE ${\displaystyle \sum_{C \in rslc(t)}
\prod_{P \in MS(t)} x^{(P)}_{s_P s_P'} \!\!\!\!
\prod_{
{
Q \in PS(t) \setminus MS(t) \atop \wedge \; \exists P \in MS(t): \; Q \in N_{must}(Sys,P,c)
}
}\!\!\!\!\!\!\! x^{(Q)}_{s_Q s_Q} \;\;\;\;
\prod_{R \in C} x^{(R)}_{s_R s_R}
= \gamma \cdot f}$
\label{line:createeq3}
\ENDFOR
\STATE solve system of equations, if successful set $found := true$
\IF{$found = true$}
\STATE $T_{mod} := T_{mod} \setminus T_c$
\FOR{each $T \in T_c$}
\STATE $factor(t) := 1$
\STATE // needed in case same transition is considered again later
%\STATE {\color{red} Update $T_{mod}$}
\ENDFOR
\label{line:originalTcEnd}
\STATE
\ELSE
%\newpage
\STATE // {\bf Algorithm Part C:}
\STATE // it still holds that $found = false$
\label{line:ddd}
% \STATE // try to involve processes from the ``pseudo'' involved set $PIS(\hat{t}) := IS(\hat{t}) \setminus IS_r(\hat{t})$
% \STATE determine $ PIS(\hat{t}) := IS(\hat{t}) \setminus IS_r(\hat{t}) = \{ P_{q_1}, \dots, P_{q_l}  \}$
%\IF{$PIS(\hat{t}) \neq \emptyset$}
% \STATE // first save the current model, in order to be able to revert to it later if necessary
% \STATE save the current model
% \label{line:savemodel}
\STATE // now try to change the model:
\label{line:extendedTcBegin}
\IF {$PS(\hat{t}) \neq IS(\hat{t})$}
\STATE // try to add $c$ to the sync.\ set of every internal node $X$
\STATE // of the form $X = X_1 ||_{\neg c} X_2$, within the current $c$-scope,
\STATE // but this must not lead to sp.\ tr.\ outgoing from a reachable state.
\STATE // If $c$ can be added, then the necessary $c$-selfloops must also be added.
\FOR {all nodes $X=X_1 ||_{\neg c}X_2$ in current $c$-scope (in bottom-up order)}
\label{line:BigForBegin}
\STATE TRYSYNC $(X,c)$
\label{line:trysync1}
\STATE // TRYSYNC
%(cf.\ App.~\ref{AppendixTRYSYNC})
tries to add $c$ to the sync.\ set of node $X$,
\STATE // if possible changes $X$ to $||_c$ and adds necessary selfloops below $X$
\ENDFOR
\label{line:BigForEnd}
\FOR{each $t := ((s_1, \dots, s_n) \xlongrightarrow{c, \gamma \cdot f} (s_1', \dots, s_n')) \in T_c$}
\STATE determine $PS(t) := \{ P_{p_1}, \dots, P_{p_k} \} \subseteq IS(\hat{t})$
\STATE find the set of relevant selfloop combinations rslc$(t) := RSLC(t,r)$
\label{line:rslc2}
\STATE // similar calculations as in the two previous lines
\STATE // have already been performed in lines~\ref{line:createeq0} and \ref{line:createeq1}, but now the
\STATE // model is changed, so $PS(t)$ will potentially be larger than before,
\STATE // and the sets in rslc$(t)$ will potentially be larger than before!
%\STATE determine $NSL(t)$
\STATE create an equation
\STATE ${\displaystyle \sum_{C \in rslc(t)}
\prod_{P \in MS(t)} x^{(P)}_{s_P s_P'}\!\!\!\!
\prod_{
{
Q \in PS(t) \setminus MS(t) \atop \wedge \; \exists P \in MS(t): \; Q \in N_{must}(Sys,P,c)}
}
\!\!\!\!\!\!\!\! x^{(Q)}_{s_Q s_Q} \;\;\;\;
\prod_{R \in C} x^{(R)}_{s_R s_R}
= \gamma \cdot f}$
\label{line:createeq4}
\label{line:NSL}
\ENDFOR
\STATE solve system of equations, if successful set $found := true$
\IF{$found = true$}
\STATE // changing the model was successful,
%\STATE // so we will {\bf not} go back to the previous
%model saved in line~\ref{line:savemodel}
%\STATE // we just have to record the fixed transitions:
\STATE $T_{mod} := T_{mod} \setminus T_c$
\FOR{each $T \in T_c$}
\STATE $factor(t) := 1$
\STATE // needed in case same transition is considered again later
% \STATE {\color{red} Update $T_{mod}$}
\ENDFOR
\label{line:extendedTcEnd}
\ELSE
\STATE // no solution found in current $c$-scope, even with added selfloops
\STATE // it still holds that $found = false$
%\STATE revert to the previous model saved in line~\ref{line:savemodel}
\STATE $curr\_root :=$ root of current $c$-scope
\ENDIF
\ENDIF
\ENDIF
\label{line:bbb}
\ENDIF
\STATE
%\newpage
\STATE // {\bf Algorithm Part D:}
\WHILE{$found = false$ and $curr\_root \neq root(Sys)$}
\label{line:moveupBegin}
\STATE // move current root upwards and try to expand $c$-scope
\STATE $curr\_root := parent(curr\_root)$
\STATE $success :=$ TRYSYNC($curr\_root$, $c$)
\label{line:trysync2}
\STATE // TRYSYNC
%(cf.\ App.~\ref{AppendixTRYSYNC})
tries to add $c$ to the sync.\ set of node $curr\_root,$
\STATE // and if possible changes $curr\_root$ to $||_c$ and adds the
\STATE // necessary selfloops below $curr\_root$
\IF{$success$}
\STATE $IS_{new}(\hat{t}) := $ all leaves of subtree rooted at $curr\_root = \{ P_{i_1}, \ldots, P_{i_m} \}$
\STATE // this $m$ is now larger than in line~\ref{line:mmm},
\STATE // it gets larger in each iteration of this while-loop
\STATE $IS(\hat{t}) := IS_{new}(\hat{t})$
\STATE $T_c := \{t \in T \; | \; action(t) = c \; \wedge \; PS(t) \cap IS(\hat{t}) \neq \emptyset \}$
\STATE $r := curr\_root$ // root of current $c$-scope
\STATE now the same code as in lines \ref{line:aaa}--\ref{line:bbb}
%(without lines \ref{line:ccc} - \ref{line:ddd})
\STATE \dots
\ELSE
\STATE BREAK while-loop of lines~\ref{line:moveupBegin}--\ref{line:moveupEnd} and move $curr\_root$ up further
\STATE // it was not possible to add $c$ to the sync.\ set at $curr\_root$
\ENDIF
%\ENDIF
\ENDWHILE
\label{line:moveupEnd}
\ENDWHILE
\end{algorithmic}
%%%%%%%%%%%%%%%%%%%%%%%%%%%%%%%%%%%%%%%%%%%%%%%%

\newpage
\section{RSLC Algorithm (Relevant Selfloop Combinations) }
\label{AppendixRSLC}

Note: RSLC is called from the main rate lifting algorithm in lines
\ref{line:createeq1} and
\ref{line:rslc2}.\\

\begin{algorithmic}[1]
\STATE {\bf Algorithm} RSLC ($t,n$)
\STATE // $t \in T$ is a transition
\STATE // $n$ is a node of the process tree of $Sys$
\STATE // The algorithm returns a set of sets of sequential processes
\STATE // (each representing a relevant selfloop combination contributing to $t$)
\IF{type($n$) = leaf}
\IF{$P_n \in (PS(t) \cap SS(t)) \setminus
\bigcup_{P \in MS(t)} N_{must}(Sys,P,c)$}
\STATE // $P_n$ denotes the process represented by leaf-node $n$
\STATE // $P_n \in PS(t)$ ensures that $P_n$ has a selfloop at its current state
\STATE return $\{ \{ P_n \} \}$
\STATE // a set containing a singleton set is returned
\ELSE
\STATE return $\{ \emptyset \}$
\STATE // the set containing the empty set is returned
\ENDIF
\ELSIF{type($n$) = $||_c$}
\STATE return $\{ C_1 \cup C_2 \; | \; C_1 \in RSLC(t,\mbox{lchild}(n))
\wedge C_2 \in RSLC(t,\mbox{rchild}(n)) \}$
\STATE // all combinations of left and right subtree
\ELSE
\STATE // it holds that type($n$) = $||_{\neg c}$
\STATE return $\{ C \; | \; C \in RSLC(t,\mbox{lchild}(n))
\vee C \in RSLC(t,\mbox{rchild}(n)) \}$
\STATE // the (disjoint) union of left and right subtree
\ENDIF
\end{algorithmic}

%%%%%%%%%%%%%%%%%%%%%%%%%
\newpage
%%%%%%%%%%%%%%%%%%%%%%%%%
%\renewcommand{\thesection}{\Alph{section}}%

\section{TRYSYNC Function }
\label{AppendixTRYSYNC}

Note: 
TRYSYNC is called from the main algorithm in lines
\ref{line:trysync1} and \ref{line:trysync2}.\\

\begin{algorithmic}[1]
\STATE {\bf Function: }TRYSYNC ($X,c$) {\bf returns} Boolean
\STATE // This function checks whether adding action $c$ to node $X$ of type $||_{\neg c}$
\STATE // is possible without creating spurious transitions.
\STATE // If possible, it changes $X$ to $||_c$ and adds the necessary selfloops.
\STATE // The return value is {\bf true} iff action $c$ could be successfully added.
\IF{changing $X$ to $||_c$ would cause sp. tr. of type A}
\label{line:notConcEnabled}
\STATE // checking for sp.\ tr.\ of type A can be done
\STATE //  by ensuring that for all reachable states $\overrightarrow{s}=(s_1,...,s_n)$,
\STATE // transitions $\overrightarrow{s}_{X_1}\xlongrightarrow{c} \dots$ and $\overrightarrow{s}_{X_2}\xlongrightarrow{c} \dots$ do not both exist, 
\STATE // where $X_1$ and $X_2$ are the children of node $X$.\footnote{$\overrightarrow{s}_{X_i}$ is a vector representing the states of the sequential components in process $X_i$, i.e.\ it is the projection of $\overrightarrow{s}$ to  $X_i$.}
\RETURN {\bf false} // do nothing, since $X$ can't be made $c$-synchronising
\ENDIF
\STATE compute $C:=COMB (X,c)$
\label{line:comb1}
\STATE // combinations of $c$-participants below $X$, 
assuming $X$ was of type $||_c$
\STATE $T^X_c=\{t \in T_c \; | \; PS(t) \cap P_X\not =\emptyset\}$
\STATE // the subset of $c$-transitions with some process from $X$ participating
\STATE // compute the feasible combinations  $C^{feas} \subseteq C$:
\STATE $C^{feas}:=\{C_i \in C \; | \; \forall t \in T^X_c:$
\label{line:feasbegin}
\STATE $\hspace{1.4cm}(\text{selfloops in all }P_k \in C_i \setminus PS(t) \text{ at state } source_k(t) \text{ are present}$
\STATE $\hspace{1.4cm} \text{or can be added without causing sp.\ tr.\ of type B)}\}$
%, see Sect.~\ref{Type B}
\label{line:feasend}
\IF {$C^{feas}=\emptyset$}
\RETURN {\bf false} // do nothing, since $X$ can't be made $c$-synchronising
\ENDIF
\STATE // now we know that $X$ can be made $c$-synchronising!
\STATE change $X$  to type $||_c$
\STATE // now permanently add the new selfloops:
\FOR {each $C_i \in C^{feas}$}
\FOR {each $t \in T^X_c$}
\STATE // it suffices to look at projection $t^X:\overrightarrow{s_X} \xlongrightarrow{c} \overrightarrow{s_X}'$
\FOR {each $P_k \in C_i \wedge P_k \not \in PS(t)$}
\IF {selfloop
in $P_k: s_k \xlongrightarrow{c} s_k$
does not exist}
\STATE add selfloop in $P_k: s_k \xlongrightarrow{c,x^{P_k}_{s_k,s_k}} s_k$
\label{line:addselfloops}
\ENDIF
\ENDFOR
\ENDFOR
\ENDFOR
\RETURN {\bf true} // $X$ has been changed to $||_c$ and the selfloops added
\STATE {\bf end function}
\end{algorithmic}

\newpage
\noindent
{\bf Checking for Spurious Transitions of Type B:}
We now discuss the details of how it can be checked whether a combination $C_i \in C = COMB(X,c)$ 
is feasible or not, i.e.\ the details of lines \ref{line:feasbegin} - \ref{line:feasend}
of function TRYSYNC.
We are in the process of checking whether making node $X=X_1||_{\neg c}X_2$ $c$-synchronising
(and adding some selfloops as needed by combination $C_i$)
causes spurious transition of type B or not.
The following program segment performs this check:
\begin{algorithmic}[1]
\STATE assume (temporarily) that $X = X_1 ||_c X_2$
\STATE // i.e.\ assume that $X$ was $c$-synchronising
\STATE // compute the set of newly needed selfloops for combination $C_i$:
\STATE $Selfloops(C_i) := \{
(P_k,s_{k_j}) \; | \;
\exists t \in T^X_c:
P_k \in C_i \setminus PS(t) \wedge s_{k_j}=source_{P_k}(t)$
\STATE $ \hspace{25mm} \wedge
\text{ selfloop } s_{k_j} \xlongrightarrow{c} s_{k_j} \text{ does not yet exist}
\} $
\FOR{all $P_k$ where any new selfloops are needed for $C_i$}
\label{line:outerforbegin}
\STATE // assume that $P_k$ is part of $X_1$, otherwise symmetric procedure
\STATE let $Y:=Y_1||_{c}Y_2$ be the lowest $c$-synchronising node containing $P_k$
\STATE // assume that $P_k$ is part of $Y_1$, otherwise symmetric procedure
\STATE // $Y$ is either part of $X_1$, or $Y=X$, or $Y$ is even above $X$ (when moving upwards)
\FOR{all states $s_{k_j}$ where a selfloop is needed and not yet present in $P_k$}
\STATE temporarily add selfloop $s_{k_j} \xlongrightarrow{c} s_{k_j}$
\STATE set $Z := P_k$ and $\overrightarrow{z} := s_{k_j}$
\STATE // we use $Z$ to denote a subsystem with a selfloop and $\overrightarrow{z}$ its state, \STATE // in order to be flexible when moving $Y$ upwards (see line~\ref{line:rerun})
\IF {$Y \neq X$ and $\exists \overrightarrow{s}$ reachable state such that $\overrightarrow{s}_{Z}=\overrightarrow{z}$ and $\exists \overrightarrow{s}_{Y_2}\xlongrightarrow{c} \ldots$}
\label{line:ifclausebegin}
\STATE // non-selfloop $c$-transition in $\overrightarrow{s}_{Y_2}$ would synchronise with
\STATE // new (atomic or combined) selfloop at $\overrightarrow{z}$
\STATE remove all selfloops newly added while processing $X$
\STATE BREAK  // a sp.\ tr.\ exists, therefore $C_i$ is not feasible
\ELSIF {$\exists \overrightarrow{s}$ reachable state such that $\overrightarrow{s}_{Z}=\overrightarrow{z}$ and $\exists \overrightarrow{s}_{Y_2}\xlongrightarrow{c} \overrightarrow{s}_{Y_2}$}
\STATE // selfloop $c$-transition in $\overrightarrow{s}_{Y_2}$ would sync
\STATE // with new (atomic or combined)  selfloop at $\overrightarrow{z}$
\STATE // yielding a new combined selfloop 
$\overrightarrow{s}_{Y}\xlongrightarrow{c} \overrightarrow{s}_{Y}$
\STATE move upwards, i.e.\ rerun the if-clause
(lines~\ref{line:ifclausebegin} - \ref{line:ifclauseend}) for $Z := Y$, $\overrightarrow{z} := \overrightarrow{s}_Y$ and $Y:=$ next higher node $Y$ of type $||_c$, if such a higher node $Y$ exists
\label{line:rerun}
\ELSE
\STATE // no $c$-transition nor $c$-selfloop in $\overrightarrow{s}_{Y_2}$
\STATE selfloop can be added to $s_{k_j}$ without causing sp.\ tr.\
\ENDIF
\label{line:ifclauseend}
\ENDFOR
\ENDFOR
\label{line:outerforend}
\STATE // once both FOR-loops have terminated without BREAK,  we know
\STATE // that $C_i$ is feasible
\end{algorithmic}

%%%%%%%%%%%%%%%%%%
\newpage
%%%%%%%%%%%%%%%%%%
\section{COMB Algorithm }
\label{AppendixCOMB}

Note: COMB is called from TRYSYNC in line
\ref{line:comb1}.\\

\begin{algorithmic}[1]
\STATE {\bf Algorithm} COMB ($X,c$)
\STATE // This algorithm returns all sequential component combinations under
\STATE // node $X$ wrt action $c$. Every combination consists of some sequential
\STATE //  components that may participate in a $c$-transition.
\STATE // Like RSLC, this algorithm returns a set of sets of sequential processes.
\IF{type(X)=leaf}
\RETURN $\{\{P_X\}\}$
\STATE // the seq. process
\ELSIF {type($X$)=$||_c$}
\RETURN 
\STATE $\{C_1 \cup C_2 \; | \; C_1 \in \text{COMB}(\text{lchild}(X),c) \wedge C_2 \in \text{COMB}(\text{rchild}(X),c)\}$
\ELSE
\STATE //type(X)=$||_{\neg c}$
\RETURN 
\STATE $\{C \; | \; C \in \text{COMB}(\text{lchild}(X),c) \vee C \in \text{COMB}(\text{rchild}(X),c)\}$
\ENDIF
\end{algorithmic}

\newpage
\section{Calculating the Participating Set}
\label{AppendixPS}
In Sec.~\ref{subsec:movandpar} quite a complicated closed-form expression for the Participating Set of a transition was given.
%{\color{red} Move this ``In practice...'' to the appendix?}
In practice, the participating set $PS(t)$ for a given transition $t$ in an SPA system can be obtained by the following procedure:
We first calculate a set of candidate processes $PS_{cand}$:
\begin{equation*}
\begin{array}{ll}
PS_{cand}(t)=& \{P_i \in SS(t)|  \\
 &(\exists P_j \in MS(t) : (P_i \in N_{may}(Sys,P_j,a) \wedge (\text{selfloop }  s_i \xrightarrow{a, \lambda_i} s_i \text{ exists} )))\\
	%\space & (P_i \in N_{may}(Sys,P_j,a) \wedge (\text{selfloop}:  s_i \xrightarrow{a, \lambda_i} s_i \text{exists} ))\\
\space & \wedge( \not\exists P_j \in MS(t): P_i \in N_{cannot}(Sys,P_j,a))\}\\
\end{array} 
\end{equation*}
Then we determine the set $PS_{may} \subseteq PS_{cand}$ as follows:
For each $P_k \in PS_{cand}(t)$, let
$r$ of type $\parallel_a$ be the root of the smallest subtree containing $P_k$ 
and at least one of the components of $MS(t)$).
For each node $N$ of type ${\parallel_a}$ on the path from $P_k$ to $r$,
excluding $r$, let $X_{N}$ be the child process of node $N$ which does not contain $P_k$.
If a selfloop $(\overrightarrow{s}_{X_N} \xrightarrow{a} \overrightarrow{s}_{X_N})$ exists in process $X_N$ 
%$(s_{P_{X_1}},\ldots,s_{P_k},\ldots s_{P_{X_g}}) \xrightarrow{a} (s'_{P_{X_1}},\ldots,s'_{P_k},\ldots s'_{P_{X_g}})$ 
then $P_k\in PS_{may}(t)$.
%where $PS_{may} \subset PS(t)$. 
Overall we get:
\begin{equation*}
PS(t)=MS(t) \; \cup \; PS_{may}(t) \; \cup \;
\{P_i \in SS(t) \;|\; \exists P_j \in MS(t) : P_i \in N_{must}(Sys, P_j, a)\}
\end{equation*}

Returning to the example from Sec.~\ref{subsec:movandpar} (Figure~\ref{PSexample}):
For transition $t=((1,1,3,1,2) \xrightarrow{a} (2,1,3,1,2))$,
we assumed that there are selfloops in state 1 in  $P_2$ and also in state 3 in $P_3$,
so both $P_2$ and $P_3$ are in $PS_{cand}(t)$.
Since on the path from $P_2$ to $r$ ($=X_1$) there is no node of type $||_a$,
nothing else needs to be checked for $P_2$, so $P_2$ is in $PS_{may}$.
For $P_3$, we need to check whether the other child of $X_3$
(which is $X_4$) can peform
an $a$-selfloop at the current state $\overrightarrow{s}_{X_4}=(1,2)$.
But since this is not the case, $P_3$ is not in $PS_{may}$
and therefore not in $PS(t)$.

%%%%%%%%%%%%%%%%%%%%%%%
\end{subappendices}
\end{document}